\documentclass[12pt,draftcls,onecolumn]{IEEEtran}
\usepackage{cite}
\usepackage{amsmath,amssymb,amsthm}    
\usepackage[dvips]{graphicx}   
\usepackage{verbatim}   
\usepackage{color}      
\usepackage{subfigure}  
\usepackage{epsfig}
\usepackage{float}
\usepackage{algorithm}
\usepackage{algorithmic}
\usepackage{setspace}
\usepackage{subfigure}
\newtheorem{lemma}{\bf Lemma}
\newtheorem{theorem}{\bf Theorem}

\makeatletter
\newcommand*{\rom}[1]{\expandafter\@slowromancap\romannumeral #1@}
\makeatother

\begin{document}
\title{\textbf{Asymptotic Analysis of Distributed Bayesian Detection with Byzantine Data}}
\author{Bhavya~Kailkhura,~\IEEEmembership{Student Member,~IEEE}, Yunghsiang~S. Han,~\IEEEmembership{Fellow,~IEEE},
Swastik~Brahma,~\IEEEmembership{Member,~IEEE}, Pramod~K.~Varshney,~\IEEEmembership{Fellow,~IEEE}
\thanks{This work was supported in part by ARO under Grant W911NF-13-2-0040 and National Science Council of Taiwan, under grants NSC 99-2221-E-011-158 -MY3, NSC 101-2221-E-011-069 -MY3. Han's work was completed during his visit to Syracuse University from 2012 to 2013.}
\thanks{B. Kailkhura, S. Brahma and P. K. Varshney are with Department of EECS, Syracuse University, Syracuse, NY 13244. (email: bkailkhu@syr.edu; skbrahma@syr.edu; varshney@syr.edu)}
\thanks{Y. S. Han is with EE Department, National Taiwan University of Science and Technology, Taiwan, R. O. C. (email: yshan@mail.ntust.edu.tw)}}
\date{}
\maketitle
\begin{abstract} 
In this letter, we consider the problem of distributed Bayesian detection in the presence of data falsifying Byzantines in the network.  
The problem of distributed detection is formulated as a binary hypothesis test at the fusion center (FC) based on 1-bit data sent by the sensors. Adopting Chernoff information as our performance metric, we study the detection performance of the system under Byzantine attack in the asymptotic regime.  
The expression for minimum attacking power required by the Byzantines to blind the FC is obtained. More specifically, we show that above a certain fraction of Byzantine attackers in the network, the detection scheme becomes completely incapable of utilizing the sensor data for detection.
When the fraction of Byzantines is not sufficient to blind the FC, we also provide closed form expressions for the optimal attacking strategies for the Byzantines that most degrade the detection performance.
\end{abstract}

\section{Introduction}
Distributed detection is a well studied topic in the detection
theory literature~\cite{Varshney, Viswanathan, veer}. 
In distributed detection systems, due to bandwidth and energy constraints, the nodes often make a 1-bit local decision regarding the presence or absence of a phenomenon before sending it to the fusion center (FC). Based on the local decisions transmitted by the nodes, the FC makes a global decision about the presence or absence of the phenomenon of interest. 
The performance of such systems strongly depends on
the reliability of the nodes in the network. The distributed nature of such systems makes them quite vulnerable
to different types of attacks. 
One typical attack on such networks is a Byzantine attack. While Byzantine attacks (originally proposed by \cite{Lamport}) may, in general, refer to many types of malicious behavior, our focus in this letter is on data-falsification attacks ~\cite{frag, Rifa, Marano, Rawat, bhavyaj, Kailkhura2013, Kailkhura, aditya}. 

Distributed detection in the presence of Byzantine
attacks has been explored in the past in~\cite{Marano,Rawat}, where the problem of determining the most effective attacking strategy of the Byzantine nodes was explored. In \cite{Marano}, the authors considered the Neyman-Pearson (NP) setup and determined the optimal attacking strategy which minimizes the detection error exponent. This approach, based on Kullback-Leibler divergence (KLD), is analytically tractable and yields approximate results in non-asymptotic cases. They also assumed that the Byzantines know the true hypothesis, which obviously is not satisfied in practice but does provide a bound. 
In \cite{Rawat}, the authors analyzed the same problem in the context of collaborative spectrum sensing under Byzantine Attacks. They relaxed the  assumption of perfect knowledge of the hypotheses by assuming that the Byzantines
determine the knowledge about the true hypotheses from their own sensing observations.
Schemes for Byzantine node identification have
been proposed in ~\cite{aditya,Rawat,a1,a2,covert}.
Our focus in this letter is considerably different from Byzantine node identification schemes in that we do not try to
authenticate the data; we determine the most effective attacking
strategies and distributed detection schemes that are robust against attacks.

All the approaches discussed so far for distributed detection schemes robust to Byzantine attacks consider distributed detection under the Neyman-Pearson (NP) setup. In contrast, we
focus on the impact of Byzantine nodes on distributed Bayesian detection,
which has not been considered in the past. Adopting Chernoff information as our performance metric, we study the performance of distributed detection systems with Byzantines in the asymptotic regime.
We are interested in answering the following questions.
\begin{itemize}
\item From the Byzantines' perspective, what is the most effective attacking strategy?
\item What is the minimum fraction of Byzantines needed to blind the FC?
\item From the FC's perspective, knowing the fraction of Byzantines in the network, or an upper
bound thereof, what is the achievable performance not knowing
the identities of compromised nodes?
\end{itemize}

The signal processing problem considered in this letter is
most similar to~\cite{Rawat}. Our results, however, are not a direct application of those
in~\cite{Rawat}. While as in~\cite{Rawat}, we
are also interested in the worst distribution pair, our objective function and, therefore, the techniques to find them are different. 
In contrast to~\cite{Rawat}, where only optimal strategies to blind the FC were obtained, we also provide closed form expressions for the optimal attacking strategies for the Byzantines that most degrade the detection performance when the fraction of Byzantines is not sufficient to blind the FC. 
Indeed, finding the optimal Byzantine
attacking strategies is only the first step toward designing a robust distributed detection system. 
The knowledge of optimal attack strategies can be further used to implement the optimal detector at the FC.

\begin{figure}[t!]
  \centering
    \includegraphics[height=0.25\textheight, width=0.4\textwidth]{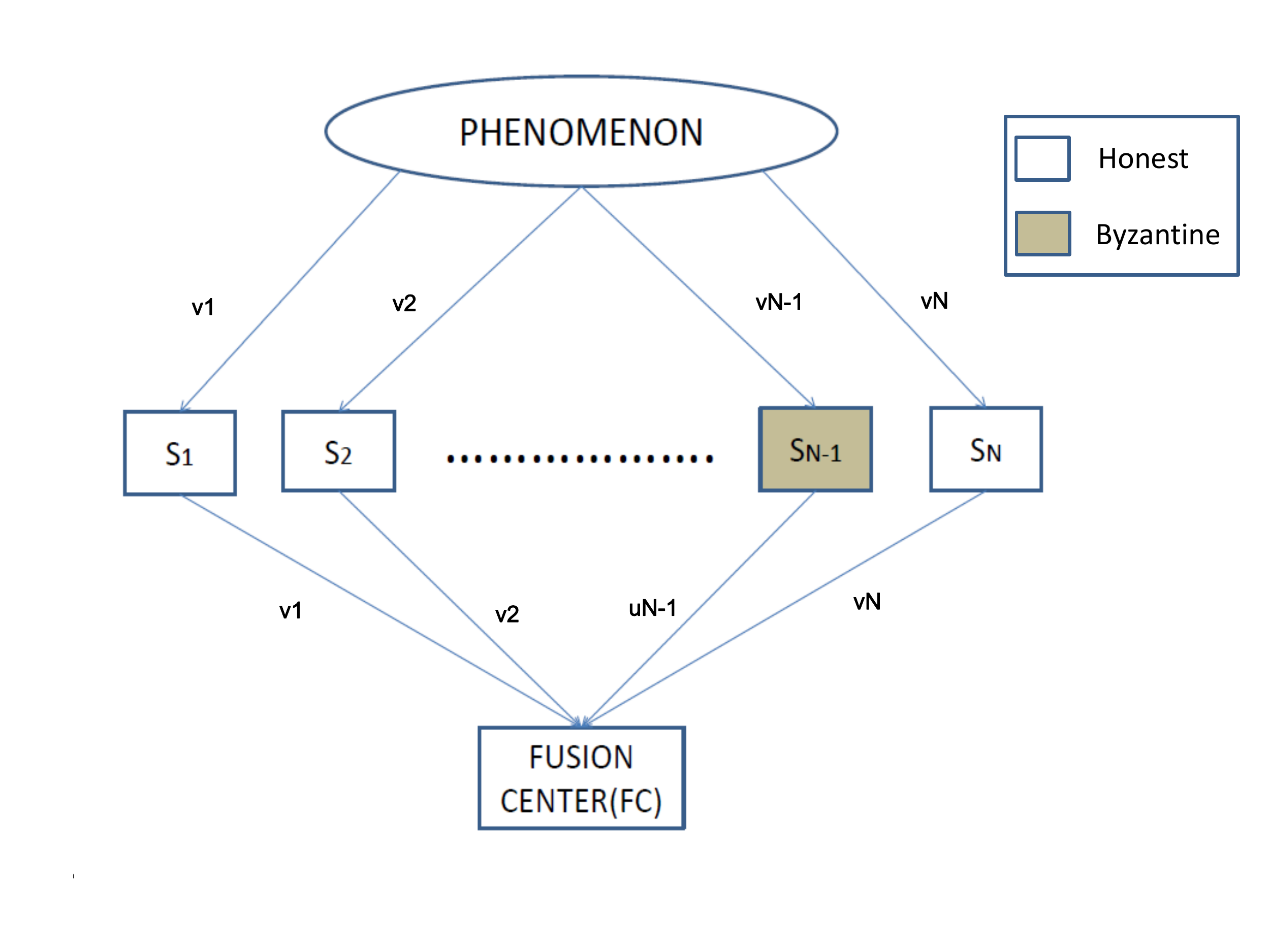}
    \caption{System Model}
    \label{model}
\end{figure}

\section{Distributed detection in the presence of Byzantines}
\label{sec2}
Consider two hypotheses $H_{0}$ (signal is absent)
and $H_{1}$ (signal is present).
Also, consider a parallel network (see Figure~\ref{model}), comprised of a central entity (known as the Fusion Center (FC))
and a set of $N$ sensors (nodes), which
faces the task of determining which of the two hypotheses is true.
Prior probabilities of the two hypotheses $H_{0}$ and $H_{1}$
are denoted by $P_{0}$ and $P_{1}$, respectively. 
The sensors observe the phenomenon, carry out local computations to decide the presence or
absence of the phenomenon, and then send their local decisions to the FC that makes a final decision
after processing the local decisions.
Observations at the nodes are assumed to be
conditionally independent and identically distributed.
A Byzantine attack on such a system compromises some of the nodes
which may then intentionally send falsified local decisions to the FC to make the final
decision incorrect.
We assume that a fraction $\alpha$ of the $N$ nodes which observe the phenomenon
have been compromised by an attacker.
We consider the communication channels to be error-free.
Next, we describe the modus-operandi  
of the nodes in detail.

\subsection{Modus Operandi of the Nodes}
\label{attack model}
Based on the observations,
each node $i$  
makes a one-bit local decision $v_{i} \in \{ 0,1\}$ regarding the absence or presence of the phenomenon
using the likelihood ratio test
    \begin{equation}
        \label{eqn1}
        \dfrac{p_{Yi}^{(1)}(y_{i})}{p_{Yi}^{(0)}(y_{i})} \quad \mathop{\stackrel{v_i = 1}{\gtrless}}_{v_i = 0} \quad \lambda
    \end{equation}
    where $\lambda$ 
    is the identical threshold\footnote{It has been shown that the use of identical thresholds 
    is asymptotically optimal~\cite{tsit}.}
       used at all the sensors and $p_{Yi}^{(k)} (y_{i})$ is the conditional probability density function (PDF) of
    observation $y_i$ under the hypothesis $H_k$, where $k=0,1$.

Each node $i$, after making its one-bit local decision $v_{i}$, sends $u_i$ to the FC, where $u_i = v_i$
if $i$ is an uncompromised (honest) node, but for a compromised (Byzantine) node $i$,
$u_i$ need not be equal to $v_i$.
We denote the probabilities of detection and false alarm of each
node $i$ in the network by $P_{d}=P(v_{i}=1|H_{1})$
and $P_{f}=P(v_{i}=1|H_{0})$, respectively,
which hold for both uncompromised nodes as well as compromised nodes.

In this letter, we assume that each Byzantine decides to attack independently
relying on its own observation and decision regarding
the presence or absence of the phenomenon.
Specifically, we define the following strategies $P_{j,1}^H$, $P_{j,0}^H$ and $P_{j,1}^B$, $P_{j,0}^B$ ($j \in \{0,1\}$)
for the honest and Byzantine nodes, respectively:\\
\vspace{-0.04in}
Honest nodes:
\begin{equation}
\label{honest}
P_{1,1}^H=1-P_{0,1}^H=P^{H}(x=1|y=1)=1
\end{equation}
\begin{equation}
P_{1,0}^H=1-P_{0,0}^H=P^{H}(x=1|y=0)=0
\end{equation}
				
\noindent
Byzantine nodes:
\begin{equation}
P_{1,1}^B=1-P_{0,1}^B=P^{B}(x=1|y=1)
\end{equation}
\begin{equation}
P_{1,0}^B=1-P_{0,0}^B=P^{B}(x=1|y=0)
\end{equation}
where $P^{H}(x=a|y=b)$ ($P^{B}(x=a|y=b)$) is the probability that an honest (Byzantine) node sends $a$ to the FC
when its actual local decision is $b$. From now
onwards, we will refer to Byzantine flipping probabilities simply by $(P_{1,0}, P_{0,1})$. We also assume that the FC is not aware of the identities of Byzantine nodes and considers each node $i$ to be Byzantine with a certain probability $\alpha$.

\subsection{Performance Criterion}
\label{Testing}
The Byzantine attacker always wants to degrade the detection performance at the FC as much
as possible; in contrast, the FC wants to maximize the detection performance. 
The detection performance at the FC in the presence of the Byzantines, however, cannot be analyzed easily for the non-asymptotic case. To gain insights into the degree to which an adversary can cause performance degradation, we consider the asymptotic regime and employ the Chernoff information~\cite{HCHER} to be the network performance metric that characterizes detection performance.

If $\mathbf{u}$ is a random vector having $N$ statistically independent and identically
distributed components, $u_i$s, under both hypotheses, the optimal detector results in error probability that obeys the asymptotics
\begin{equation}
\lim_{N \rightarrow \infty} \frac{\ln P_E}{N}=-C(\pi_{1,1},\pi_{1,0}),
\end{equation}
where the Chernoff information $C$ is defined as 
\begin{eqnarray}
&&
C=\max_{0 \leq t\leq 1} -\ln (\sum_{j \in \{0,1\}} \pi_{j0}^t \pi_{j1}^{1-t}). \label{chereq} 
\end{eqnarray}
$\pi_{j0}$ and $\pi_{j1}$ in \eqref{chereq} are the conditional probabilities of $u_i=j$ given $H_0$ and $H_1$, respectively.
 Specifically, $\pi_{1,0}$ and $\pi_{1,1}$ can be calculated as
   \begin{equation}
   \label{equ1}
\pi_{1,0}=\alpha(P_{1,0}(1-P_f)+(1-P_{0,1})P_f)+(1-\alpha)P_f
\end{equation}
and
\begin{equation}
\label{equ2}
\pi_{1,1}=\alpha(P_{1,0}(1-P_d)+(1-P_{0,1})P_d)+(1-\alpha)P_d,
\end{equation}
where $\alpha$ is the fraction of Byzantine nodes.

From the Byzantine attacker's point of view, our goal is to find $P_{1,0}$ and $P_{0,1}$ that minimize Chernoff information $C$ for a given value of $\alpha$. Observe that, when $\alpha\geq 0.5$, Chernoff information can be minimized by simply making posterior probabilities equal to prior probabilities (we discuss this in more detail later in the letter). However, for $\alpha<0.5$, a closed form expression for Chernoff information is needed to find $P_{1,0}$ and $P_{0,1}$ that minimize $C$. 

\section{Closed Form Expression for the Chernoff Information}

In this section, we derive a closed form expression for the Chernoff information, when $\alpha<0.5$.\footnote{Similar results can be derived for $\alpha \geq 0.5$.}
To obtain the closed form expression for Chernoff information,
the solution of an optimization problem is required:  
$\underset{0\leq t\leq 1}{max}-\ln(\sum_{j \in \{0,1\}} \pi_{j0}^t \pi_{j1}^{1-t})$. This is easy to evaluate numerically because $(\sum_{j \in \{0,1\}} \pi_{j0}^t \pi_{j1}^{1-t})$ is convex in $t$. However, obtaining a closed form solution for this optimization problem can be tedious. Fortunately, we can find a closed form expression for the Chernoff information for $\alpha<0.5$.

\begin{lemma}
For $\alpha<0.5$, the Chernoff information between the distributions $\pi_{1,0}$ and $\pi_{1,1}$ (as given in \eqref{equ1} and \eqref{equ2}, respectively) is given by $C= -\ln (\sum_{j \in \{0,1\}} \pi_{j0}^{t^*} \pi_{j1}^{1-t^*})$ with
\begin{equation}
t^*=\frac{\ln \left( \dfrac{\ln (\pi_{1,1}/\pi_{1,0})}{\ln ((1-\pi_{1,0})/(1-\pi_{1,1}))} \dfrac{\pi_{1,1}}{1-\pi_{1,1}}\right)}{\ln \left(\dfrac{(1/\pi_{1,0})-1}{(1/\pi_{1,1})-1}\right)}.
\end{equation}
\end{lemma}
\begin{proof}
Observe that the problem of finding the optimal $t^*$ in \eqref{chereq} is equivalent to
\begin{equation}
\label{che}
\min_{0 \leq t\leq 1} \ln (\sum_{j \in \{0,1\}} \pi_{j0}^t \pi_{j1}^{1-t})
\end{equation}
which is a constrained minimization problem. To find $t^*$, we first perform unconstrained minimization (no constraint on the value of $t$) and later show that the solution of the unconstrained optimization problem is the same as the solution of the constrained optimization problem. In other words, the optimal $t^*$ is the same for both cases.   

By observing that logarithm is an increasing function, the optimization problem as given in \eqref{che} is equivalent to
\begin{equation}
\label{che-1}
\min_{t} [\pi_{1,0}^t \pi_{1,1}^{1-t}+(1-\pi_{1,0})^t (1-\pi_{1,1})^{1-t}].
\end{equation}
Now, performing the first derivative test, we have
\begin{eqnarray}
&&\frac{d}{dt}[\pi_{1,0}^t \pi_{1,1}^{1-t}+(1-\pi_{1,0})^t (1-\pi_{1,1})^{1-t}]\nonumber\\
&=&(1-\pi_{1,1})\left(\frac{1-\pi_{1,0}}{1-\pi_{1,1}}\right)^t \ln \left(\frac{1-\pi_{1,0}}{1-\pi_{1,1}}\right)\nonumber\\
&&+\pi_{1,1}\left(\frac{\pi_{1,0}}{\pi_{1,1}}\right)^t \ln \left(\frac{\pi_{1,0}}{\pi_{1,1}}\right)\label{Chernoff}.
\end{eqnarray}
The first derivative \eqref{Chernoff} is set to zero to find the critical points of the function: 
\begin{equation}
\label{chernoffK}
\left(\frac{(1/\pi_{1,0})-1}{(1/\pi_{1,1})-1}\right)^t=\frac{\ln (\pi_{1,1}/\pi_{1,0})}{\ln ((1-\pi_{1,0})/(1-\pi_{1,1}))} \left(\frac{\pi_{1,1}}{1-\pi_{1,1}}\right).
\end{equation} 
After some simplification, $t^*$ which satisfies \eqref{chernoffK} turns out to be
\begin{equation}
\label{t}
t^*=\frac{\ln \left( \dfrac{\ln (\pi_{1,1}/\pi_{1,0})}{\ln ((1-\pi_{1,0})/(1-\pi_{1,1}))} \dfrac{\pi_{1,1}}{1-\pi_{1,1}}\right)}{\ln \left(\dfrac{(1/\pi_{1,0})-1}{(1/\pi_{1,1})-1}\right)}.
\end{equation}

To determine whether the critical point is a minimum or a maximum, we perform the second derivative test. Since 
\begin{eqnarray}
&&\frac{d^2}{d^2t}[\pi_{1,0}^t \pi_{1,1}^{1-t}+(1-\pi_{1,0})^t (1-\pi_{1,1})^{1-t}]\nonumber\\
&=&(1-\pi_{1,1})\left(\frac{1-\pi_{1,0}}{1-\pi_{1,1}}\right)^t \left(\ln \frac{1-\pi_{1,0}}{1-\pi_{1,1}}\right)^2\nonumber\\
&&+\pi_{1,1}\left(\frac{\pi_{1,0}}{\pi_{1,1}}\right)^t \left(\ln \frac{\pi_{1,0}}{\pi_{1,1}}\right)^2
\end{eqnarray}
 is greater than zero, $t^*$ as given in \eqref{t} minimizes \eqref{che-1}. Since $0\leq t^*\leq 1$ (See proof in Appendix \ref{proof1}), $t^*$ as given in \eqref{t} is also the solution of \eqref{che}. 
\end{proof}

\section{Asymptotic Analysis of Optimal Byzantine Attack}
\label{sec4}
 
 First, we will determine the minimum fraction of Byzantines needed to blind the decision fusion scheme.
 
\subsection{Critical Power to Blind the Fusion Center}
\label{sec3}
In this section, we determine the minimum fraction of Byzantine nodes needed to make the FC ``blind'' and denote it by $\alpha_{blind}$. We say that the FC is blind if an adversary can make the data that the FC receives from the sensors such that no information is conveyed. In other words, the optimal detector at the FC cannot perform better than simply making the decision based on priors. 
\begin{lemma}
In Bayesian distributed detection, the minimum fraction of Byzantines needed to make the FC blind is $\alpha_{blind}=0.5$. 
\end{lemma}
\begin{proof}
The FC becomes blind if the probability of receiving a given vector $\mathbf{u}$ is independent of the hypothesis present. Using the conditional i.i.d. assumption, under which observations at the nodes are conditionally independent and identically
distributed, the condition to make the FC blind becomes $\pi_{1,1}=\pi_{1,0}$. This is true only when
\begin{equation*}
\alpha[P_{1,0}(P_f-P_d)+(1-P_{0,1})(P_d-P_f)]+(1-\alpha)(P_d-P_f)=0.
\end{equation*}
Hence, the FC becomes blind if
\begin{equation}
\label{blind}
\alpha=\dfrac{1}{(P_{1,0}+P_{0,1})}.
\end{equation}
$\alpha$ in \eqref{blind} is minimized when $P_{1,0}$  and $P_{0,1}$ both take their largest values, i.e., $P_{1,0}=P_{0,1}=1$.
Hence, $\alpha_{blind}=0.5$.
\end{proof}

Next, we find the optimal attacking strategies which minimize the Chernoff information.

\subsection{Minimization of Chernoff Information}
 First, we minimize Chernoff information for $\alpha<0.5$. Later in the section, we generalize our results for any arbitrary $\alpha$. 
Since logarithm is an increasing function, the problem of minimizing the Chernoff information is equivalent to the following problem:
\begin{figure*}[t]
\centering
\subfigure[] {
\includegraphics[height=0.25\textheight, width=0.4\textwidth]{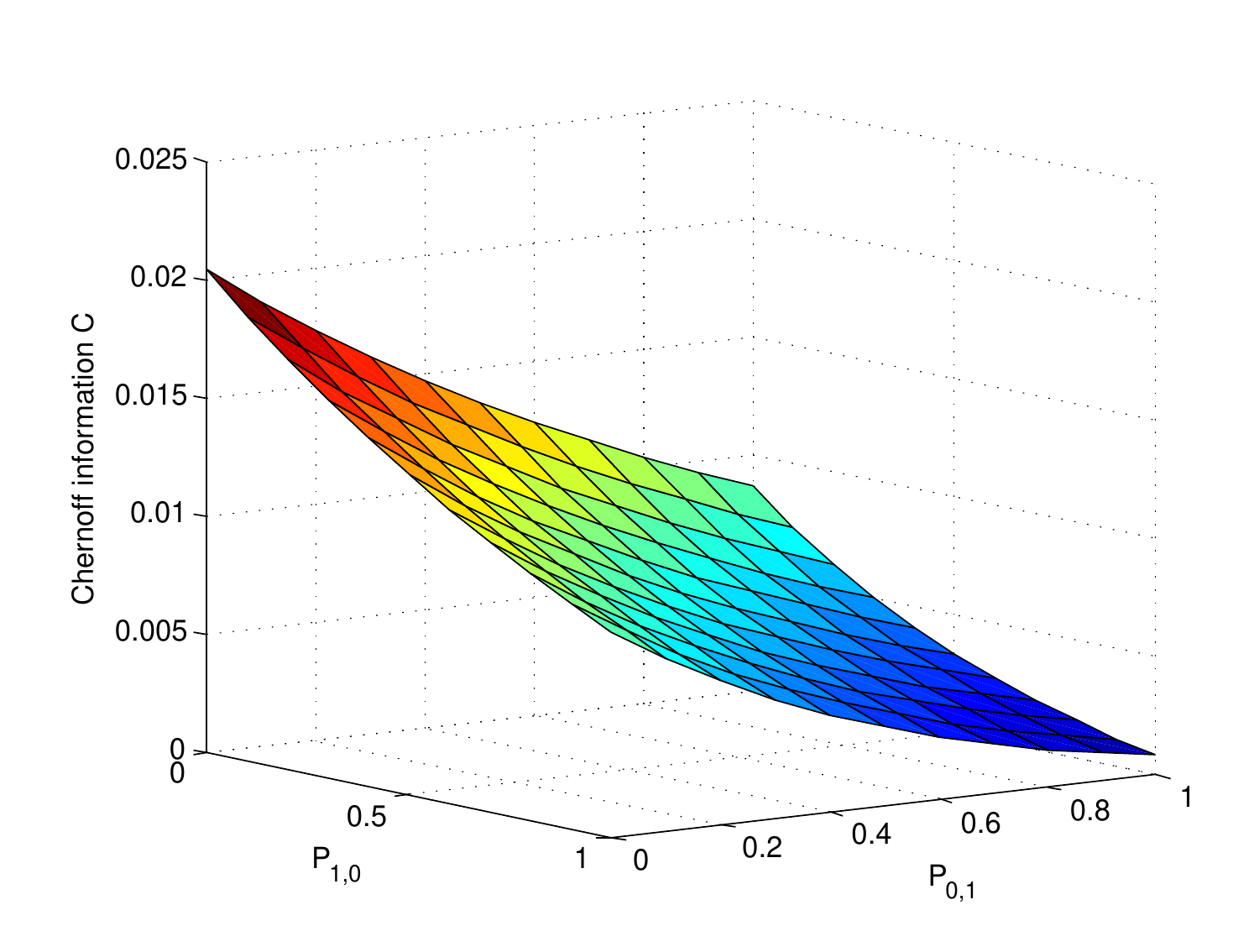}
\label{cher} }
\subfigure[]{
\includegraphics[height=0.25\textheight, width=0.4\textwidth]{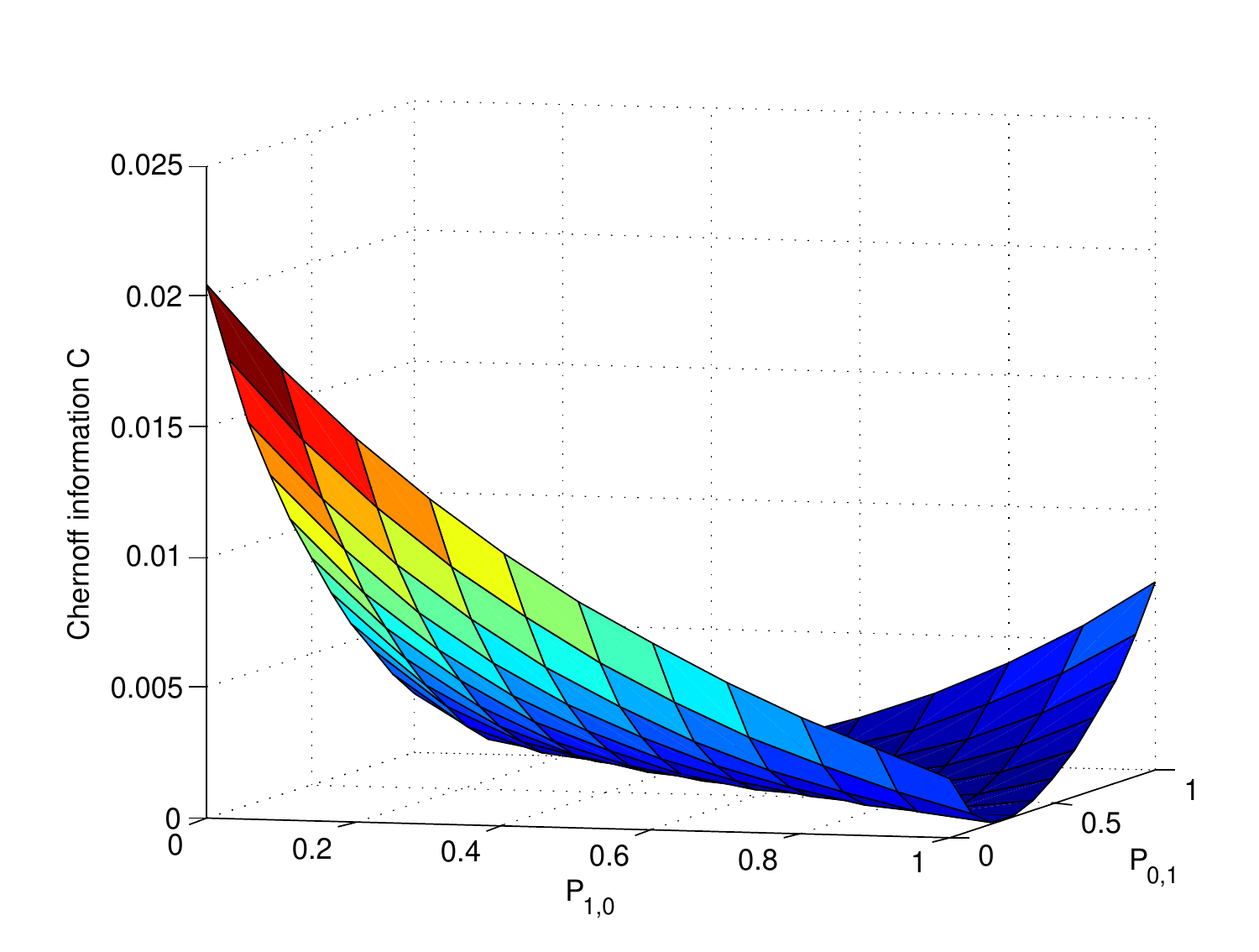}
\label{cher1}}
\caption{ \subref{cher} Chernoff information as a function of $(P_{1,0},P_{0,1})$ for $\alpha=0.4$. \subref{cher1} Chernoff information as a function of $(P_{1,0},P_{0,1})$ for $\alpha=0.8$.}
\label{local}
\end{figure*} 
\begin{equation*}
\begin{aligned}
& \underset{P_{1,0},P_{0,1}}{\text{maximize}}
& & \pi_{1,0}^{t^*} \pi_{1,1}^{1-{t^*}}+(1-\pi_{1,0})^{t^*} (1-\pi_{1,1})^{1-{t^*}} \\
& \text{subject to}
& &  0 \leq P_{1,0}\leq 1\\
& & & 0 \leq P_{0,1}\leq 1\\
\end{aligned}
\tag{P1}\label{opt-P}
\end{equation*}
where $\alpha<0.5$ and $t^*$ is as given in \eqref{t}. 

Let us denote $\tilde{C}=\pi_{1,0}^{t^*} \pi_{1,1}^{1-{t^*}}+(1-\pi_{1,0})^{t^*} (1-\pi_{1,1})^{1-{t^*}}$. Observe that, maximization of $\tilde{C}$ is equivalent to the minimization of Chernoff information $C$. Next, in Lemma~\ref{Lemma-6} we present the properties of Chernoff information $C$ (for the case when $\alpha<0.5$) with respect to $(P_{1,0},P_{0,1})$ that enable us to find the optimal attacking strategies in this case.

\begin{lemma}
\label{Lemma-6}
Let $\alpha<0.5$ and assume that the optimal $t^*$ is used in the expression for the Chernoff information. Then, the Chernoff information, $C$, is a monotonically decreasing function of $P_{1,0}$ for a fixed $P_{0,1}$. Conversely, the Chernoff information is also a monotonically decreasing function of $P_{0,1}$ for a fixed $P_{1,0}$.
\end{lemma}
\begin{proof}
See Appendix \ref{proof2}.
\end{proof}
Next, using Lemma~\ref{Lemma-6}, we present the optimal attacking strategies $P_{1,0}$ and $P_{0,1}$ that minimize the Chernoff information, $C$, for $0\leq \alpha\leq 1$.

\begin{theorem}
\label{th1}
The optimal attacking strategy, $(P_{1,0}^*, P_{0,1}^*)$, which minimizes the Chernoff information is
\[ (P_{1,0}^*, P_{0,1}^*)  \left\{ \begin{array}{rll}
				(p_{1,0}, p_{0,1})  & \mbox{if}\ \alpha \geq 0.5 \\
			   (1,1)  & \mbox{if}\ \alpha < 0.5
				\end{array}\right. ,
\] 
where, $(p_{1,0}, p_{0,1})$ satisfy $\alpha(p_{1,0}+p_{0,1})=1$.
\end{theorem}
\begin{proof}
The minimum value of C is zero and it occurs when $\pi_{1,1}=\pi_{1,0}$. By \eqref{equ1} and \eqref{equ2}, $\pi_{1,1}=\pi_{1,0}$ implies
\begin{equation}
\alpha (P_{1,0}+P_{0,1})=1.\label{sum-to-one}
\end{equation}
From \eqref{sum-to-one}, when $\alpha\geq 0.5$, the attacker can always find flipping probabilities  that  make the Chernoff information equal to zero. When $\alpha=0.5$, $P_{1,0}=P_{0,1}=1$ is the optimal strategy. When $\alpha>0.5$, any pair which satisfies $P_{1,0}+P_{0,1}=\frac{1}{\alpha}$ is the optimal strategy. However, when $\alpha<0.5$, \eqref{sum-to-one} cannot be satisfied or in other words Byzantines can not make $C=0$ since $\pi_{1,1}$ can not be made equal to $\pi_{1,0}$. From Lemma \ref{Lemma-6}, when $\alpha<0.5$, the optimal attacking strategy, $(P_{1,0}, P_{0,1})$, that minimizes the Chernoff information is $(1,1)$.
\end{proof}

Next, to gain insights into Theorem~\ref{th1}, we present some illustrative examples that corroborate our results.

\subsection{Illustrative Examples}
In Figure~\ref{cher}, we plot the Chernoff information as a function of $(P_{1,0},P_{0,1})$ for $(P_d=0.6,P_f=0.4)$ and $\alpha=0.4$. It can be observed that for a fixed $P_{0,1}$ ($P_{1,0}$) the Chernoff information $C$ is a monotonically decreasing function of $P_{1,0}$, $P_{0,1}$ (as has been shown in Lemma~\ref{Lemma-6}). In other words, when $\alpha=0.4$, the attacking strategy, $(P_{1,0},P_{0,1})$, that minimizes the Chernoff information $C$ is $(1,1)$.

Similarly, in Figure~\ref{cher1}, we consider the scenario when the fraction of Byzantines in the network is $\alpha=0.8$. It can be seen from Figure~\ref{cher1} that the minimum value of the Chernoff information in this case is $C=0$. Notice that, the attacking strategy, $(P_{1,0},P_{0,1})$ that makes $C=0$ is not unique in this case. It can be verified that any attacking strategy which satisfies $P_{1,0}+P_{0,1}=\frac{1}{0.8}$ would make $C=0$. Thus, results presented in Figures~\ref{cher} and~\ref{cher1} corroborate our theoretical result presented in Theorem~\ref{th1}.
\section{Discussion and Future Work}
\label{sec:con}
We considered the problem of distributed Bayesian detection with Byzantine data, and characterized the power of attack analytically. We obtained closed form expressions for the optimal attacking strategies
that most degrade the detection performance.
The knowledge of optimal attack strategies can be further used to implement the optimal detector at the FC.
Also in addition, if only an
upper bound $\tilde{\alpha}$ on $\alpha$ is available to the FC, then, optimal attack strategies should be simply computed
using the upper bound $\tilde{\alpha}$. For any $\alpha \leq \tilde{\alpha}$, the test designed with $\tilde{\alpha}$
achieves an exponent no smaller than $C(\tilde{\alpha})$. In the future, we plan to extend our analysis to the non-asymptotic case.


\appendices
\section{Proof of $0\leq t^*\leq 1$}
\label{proof1}
 First, we show that $t^*\leq 1$. We start from the following equality:
 \begin{equation}
  \dfrac{\pi_{1,1}}{\pi_{1,0}}-1=\left(\frac{1-\pi_{1,0}}{\pi_{1,0}}-\frac{1-\pi_{1,1}}{\pi_{1,0}}\right)= \dfrac{1-\pi_{1,0}}{\pi_{1,0}}\left(1-\dfrac{1-\pi_{1,1}}{1-\pi_{1,0}}\right).\label{eq-1}
 \end{equation}
By applying the logarithm inequality $1-\dfrac{1}{x}<\ln(x)<(x-1),\; \forall x>0$,  to \eqref{eq-1}, we have 
 \begin{eqnarray*}
 \ln\dfrac{\pi_{1,1}}{\pi_{1,0}}&<& \dfrac{\pi_{1,1}}{\pi_{1,0}}-1\\
 &=&  \dfrac{1-\pi_{1,0}}{\pi_{1,0}}\left(1-\dfrac{1-\pi_{1,1}}{1-\pi_{1,0}}\right)\\ 
 &\leq& \dfrac{1-\pi_{1,0}}{\pi_{1,0}}\ln\dfrac{1-\pi_{1,0}}{1-\pi_{1,1}}.
 \end{eqnarray*}
Now,
\begin{eqnarray*} 
&&
 \ln\dfrac{\pi_{1,1}}{\pi_{1,0}}\leq \dfrac{1-\pi_{1,0}}{\pi_{1,0}}\ln\dfrac{1-\pi_{1,0}}{1-\pi_{1,1}}\\
&\Leftrightarrow&
 {\dfrac{\ln (\pi_{1,1}/\pi_{1,0})}{\ln ((1-\pi_{1,0})/(1-\pi_{1,1}))} \dfrac{\pi_{1,1}}{1-\pi_{1,1}}}\leq { \dfrac{(1/\pi_{1,0})-1}{(1/\pi_{1,1})-1}}\\
  &\Leftrightarrow&
 \frac{\ln \left( \dfrac{\ln (\pi_{1,1}/\pi_{1,0})}{\ln ((1-\pi_{1,0})/(1-\pi_{1,1}))} \dfrac{\pi_{1,1}}{1-\pi_{1,1}}\right)}{\ln \left(\dfrac{(1/\pi_{1,0})-1}{(1/\pi_{1,1})-1}\right)}\leq 1\\
   &\Leftrightarrow&
t^*\le 1.
 \end{eqnarray*}

Next, we show that $t^*\geq 0$. First we prove that the denominator of $t^*$ is positive. Since $\pi_{1,1}>\pi_{1,0}$ for $P_d>P_f$ and $\alpha<0.5$, we have
  \begin{eqnarray}
 &&
\pi_{1,1}>\pi_{1,0}\\
 &\Leftrightarrow&
 { \dfrac{(1/\pi_{1,0})-1}{(1/\pi_{1,1})-1}}> 1\\
 &\Leftrightarrow&
 {\ln \left(\dfrac{(1/\pi_{1,0})-1}{(1/\pi_{1,1})-1}\right)}> 0.
 \end{eqnarray}
 
Next we prove that the numerator of $t^*$ is nonnegative, and then $t^*$ is nonnegative.
We start from the following equality:
 \begin{equation}
  1-\dfrac{\pi_{1,0}}{\pi_{1,1}}=\left(\frac{1-\pi_{1,0}}{\pi_{1,1}}-\frac{1-\pi_{1,1}}{\pi_{1,1}}\right)= \dfrac{1-\pi_{1,1}}{\pi_{1,1}}\left(\dfrac{1-\pi_{1,0}}{1-\pi_{1,1}}-1\right).\label{eq-2}
 \end{equation} 
By applying the logarithm inequality $1-\dfrac{1}{x}<\ln(x)<(x-1),\; \forall x>0$,  to \eqref{eq-2}, we have 
 \begin{eqnarray*}
 \ln\dfrac{\pi_{1,1}}{\pi_{1,0}}&>& 1-\dfrac{\pi_{1,0}}{\pi_{1,1}}\\
 &=&\dfrac{1-\pi_{1,1}}{\pi_{1,1}}\left(\dfrac{1-\pi_{1,0}}{1-\pi_{1,1}}-1\right)\\
 &\geq& \dfrac{1-\pi_{1,1}}{\pi_{1,1}}\ln\dfrac{1-\pi_{1,0}}{1-\pi_{1,1}}.
 \end{eqnarray*}
Now,
 \begin{eqnarray*}
 &&
 \ln\dfrac{\pi_{1,1}}{\pi_{1,0}}\geq \dfrac{1-\pi_{1,1}}{\pi_{1,1}}\ln\dfrac{1-\pi_{1,0}}{1-\pi_{1,1}}\\
&\Leftrightarrow&
 { \dfrac{\ln (\pi_{1,1}/\pi_{1,0})}{\ln ((1-\pi_{1,0})/(1-\pi_{1,1}))} \dfrac{\pi_{1,1}}{1-\pi_{1,1}}}\geq 1\\
 &\Leftrightarrow&
 {\ln \left( \dfrac{\ln (\pi_{1,1}/\pi_{1,0})}{\ln ((1-\pi_{1,0})/(1-\pi_{1,1}))} \dfrac{\pi_{1,1}}{1-\pi_{1,1}}\right)}\geq 0.
 \end{eqnarray*}

\section{Proof of Lemma~\ref{Lemma-6}}
\label{proof2}
To show that, for the optimal $t^*$ and $\alpha<0.5$, Chernoff information, $C$, is monotonically decreasing function of $P_{1,0}$ while keeping $P_{0,1}$ fixed is equivalent to showing that  $\tilde{C}$, is monotonically increasing function of $P_{1,0}$ while keeping $P_{0,1}$ fixed. 
Differentiating both sides of $\tilde{C}$ with respect to $P_{1,0}$, we get
\begin{eqnarray*}
\dfrac{d\tilde{C}}{P_{1,0}}&=& \pi_{1,0}^{t^*}\pi_{1,1}^{(1-t^*)}\left(\dfrac{dt^*}{P_{1,0}}\ln \dfrac{\pi_{1,0}}{\pi_{1,1}}+(1-t^*)\dfrac{\pi_{1,1}'}{\pi_{1,1}}+t^*\dfrac{\pi_{1,0}'}{\pi_{1,0}}\right)\\
&+&(1-\pi_{1,0})^{t^*}(1-\pi_{1,1})^{(1-t^*)}\left(\dfrac{dt^*}{P_{1,0}}\ln \dfrac{1-\pi_{1,0}}{1-\pi_{1,1}}-(1-t^*)\dfrac{\pi_{1,1}'}{1-\pi_{1,1}}-t^*\dfrac{\pi_{1,0}'}{1-\pi_{1,0}}\right)
\end{eqnarray*}
In the above equation,
\begin{small}
\begin{eqnarray*}
\dfrac{dt^*}{P_{1,0}}&=& \dfrac{\left(\ln\dfrac{\pi_{1,1}}{\pi_{1,0}}+\ln\dfrac{1-\pi_{1,0}}{1-\pi_{1,1}}\right)\left(\dfrac{G'}{G}+\dfrac{\pi_{1,1}'}{\pi_{1,1}}+\dfrac{\pi_{1,1}'}{1-\pi_{1,1}}\right)-\left(\ln G+\ln\dfrac{\pi_{1,1}}{1-\pi_{1,1}}\right)\left(\dfrac{\pi_{1,1}'}{\pi_{1,1}}-\dfrac{\pi_{1,0}'}{\pi_{1,0}}+\dfrac{\pi_{1,1}'}{1-\pi_{1,1}}-\dfrac{\pi_{1,0}'}{1-\pi_{1,0}}\right)}{\left(\ln\dfrac{\pi_{1,1}}{\pi_{1,0}}+\ln\dfrac{1-\pi_{1,0}}{1-\pi_{1,1}}\right)^2}
\end{eqnarray*}
\end{small}
where $G=\dfrac{\ln(\pi_{1,1}/\pi_{1,0})}{\ln((1-\pi_{1,0})/(1-\pi_{1,1}))}$ and 
\begin{eqnarray*}
\dfrac{G'}{G}&=& \dfrac{\ln\dfrac{1-\pi_{1,0}}{1-\pi_{1,1}}\left(\dfrac{\pi_{1,1}'}{\pi_{1,1}}-\dfrac{\pi_{1,0}'}{\pi_{1,0}}\right)-\ln\dfrac{\pi_{1,1}}{\pi_{1,0}}\left(\dfrac{\pi_{1,1}'}{1-\pi_{1,1}}-\dfrac{\pi_{1,0}'}{1-\pi_{1,0}}\right)}{\ln\dfrac{\pi_{1,1}}{\pi_{1,0}} \ln\dfrac{1-\pi_{1,0}}{1-\pi_{1,1}}}.
\end{eqnarray*}
Let us denote $a_1=\ln G+\ln(\pi_{1,1}/(1-\pi_{1,1}))$, $a_2=\ln(\pi_{1,1}/\pi_{1,0})+\ln((1-\pi_{1,0})/(1-\pi_{1,1}))$, $b_1=(\pi_{1,1}'/\pi_{1,1})+(\pi_{1,1}'/(1-\pi_{1,1}))$, $b_2=(\pi_{1,0}'/\pi_{1,0})+(\pi_{1,0}'/(1-\pi_{1,0}))$, $c_1=\pi_{1,0}^{t^*}\pi_{1,1}^{1-t^*}\ln(\pi_{1,1}/\pi_{1,0})$,
$c_2=(1-\pi_{1,0})^{t^*}(1-\pi_{1,1})^{1-t^*}\ln((1-\pi_{1,0})/(1-\pi_{1,1}))$, $d_1=((1-t^*)(\pi_{1,1}'/\pi_{1,1})+t^*(\pi_{1,0}/\pi_{1,0}))\pi_{1,0}^{t^*}\pi_{1,1}^{1-t^*}$ and $d_2=((1-t^*)(\pi_{1,1}'/(1-\pi_{1,1}))+t^*(\pi_{1,0}/(1-\pi_{1,0})))(1-\pi_{1,0})^{t^*}(1-\pi_{1,1})^{1-t^*}$.
Now, $\tilde{C}$, is monotonically increasing function of $P_{1,0}$ while keeping $P_{0,1}$ fixed if
\begin{eqnarray*}
&&
a_1[b_1c_1+b_2c_2]+a_2[-(G'/G)c_1+b_1c_2]+a_2^2d_1>a_1[b_1c_2+b_2c_1]+a_2[-(G'/G)c_2+b_1c_1]+a_2^2d_2\\
&\Leftrightarrow&
a_2^2(d_1-d_2)>(c_1-c_2)(a_1(b_2-b_1)+a_2((G'/G)+b_1))\\
&\Leftrightarrow&
a_2^2(d_1-d_2)>0
\end{eqnarray*}
where the last inequality follows from the fact that $(c_1-c_2)=0$ as given in \eqref{chernoffK}.


Now, to show that $\dfrac{d\tilde{C}}{P_{1,0}}>0$ is equivalent to show that $(d_1-d_2)>0$. In other words,
\begin{equation}
{t^*}(1-P_f)\left[\left(\dfrac{\pi_{1,1}}{\pi_{1,0}}\right)^{1-{t^*}}-\left(\dfrac{1-\pi_{1,1}}{1-\pi_{1,0}}\right)^{1-{t^*}} \right]
+(1-{t^*})(1-P_d)\left[\left(\dfrac{\pi_{1,0}}{\pi_{1,1}}\right)^{{t^*}}-\left(\dfrac{1-\pi_{1,0}}{1-\pi_{1,1}}\right)^{t^*} \right]> 0.\label{lemma-2-condition1}
\end{equation}
Note that,
\begin{equation*}
\left[\left(\dfrac{\pi_{1,1}}{\pi_{1,0}}\right)^{1-{t^*}}-\left(\dfrac{1-\pi_{1,1}}{1-\pi_{1,0}}\right)^{1-{t^*}} \right]\geq 0;\;\left[\left(\dfrac{\pi_{1,0}}{\pi_{1,1}}\right)^{t^*}-\left(\dfrac{1-\pi_{1,0}}{1-\pi_{1,1}}\right)^{t^*} \right]\leq 0.
\end{equation*}
Hence, \eqref{lemma-2-condition1} can be simplified to,
\begin{equation}
\label{chernoffcond1}
\dfrac{1-P_f}{1-P_d}> \dfrac{(1-{t^*})\left[\left(\dfrac{1-\pi_{1,0}}{1-\pi_{1,1}}\right)^{t^*}-\left(\dfrac{\pi_{1,0}}{\pi_{1,1}}\right)^{t^*} \right]}{{t^*}\left[\left(\dfrac{\pi_{1,1}}{\pi_{1,0}}\right)^{1-{t^*}}-\left(\dfrac{1-\pi_{1,1}}{1-\pi_{1,0}}\right)^{1-{t^*}} \right]}.
\end{equation} 
Similarly, for the optimal ${t^*}$ and $\alpha<0.5$, Chernoff information, $C$, is monotonically decreasing function of $P_{0,1}$ while keeping $P_{1,0}$ fixed if $(d_1-d_2)>0$,
which is equivalent to show that,
\begin{equation}
{t^*}(-P_f)\left[\left(\dfrac{\pi_{1,1}}{\pi_{1,0}}\right)^{1-{t^*}}-\left(\dfrac{1-\pi_{1,1}}{1-\pi_{1,0}}\right)^{1-{t^*}} \right]
+(1-{t^*})(-P_d)\left[\left(\dfrac{\pi_{1,0}}{\pi_{1,1}}\right)^{t^*}-\left(\dfrac{1-\pi_{1,0}}{1-\pi_{1,1}}\right)^{t^*} \right]> 0.\label{lemma-2-condition2}
\end{equation}
Furthermore, \eqref{lemma-2-condition2} can be simplified to
\begin{equation}
\label{chernoffcond2}
\dfrac{P_f}{P_d}< \dfrac{(1-{t^*})\left[\left(\dfrac{1-\pi_{1,0}}{1-\pi_{1,1}}\right)^{t^*}-\left(\dfrac{\pi_{1,0}}{\pi_{1,1}}\right)^{t^*} \right]}{{t^*}\left[\left(\dfrac{\pi_{1,1}}{\pi_{1,0}}\right)^{1-{t^*}}-\left(\dfrac{1-\pi_{1,1}}{1-\pi_{1,0}}\right)^{1-{t^*}} \right]}.
\end{equation}
Combining  \eqref{chernoffcond1} and \eqref{chernoffcond2}, the condition to make Lemma \ref{Lemma-6} true becomes
\begin{equation}
\label{chernoffcond}
\dfrac{P_f}{P_d}< \dfrac{(1-{t^*})\left[\left(\dfrac{1-\pi_{1,0}}{1-\pi_{1,1}}\right)^{t^*}-\left(\dfrac{\pi_{1,0}}{\pi_{1,1}}\right)^{t^*} \right]}{{t^*}\left[\left(\dfrac{\pi_{1,1}}{\pi_{1,0}}\right)^{1-{t^*}}-\left(\dfrac{1-\pi_{1,1}}{1-\pi_{1,0}}\right)^{1-{t^*}} \right]}< \dfrac{1-P_f}{1-P_d}.
\end{equation}
Note that right hand inequality in \eqref{chernoffcond} can be rewritten as
\begin{eqnarray*} 
&& \left(\dfrac{1}{t^*}-1\right)\left[\left(\dfrac{1-\pi_{1,0}}{1-\pi_{1,1}}\right)^{t^*}-\left(\dfrac{\pi_{1,0}}{\pi_{1,1}}\right)^{t^*} \right] <  
\dfrac{1-P_f}{1-P_d} \left[\left(\dfrac{\pi_{1,1}}{\pi_{1,0}}\right)^{1-{t^*}}-\left(\dfrac{1-\pi_{1,1}}{1-\pi_{1,0}}\right)^{1-{t^*}} \right] \\
&\Leftrightarrow  & \left(\dfrac{1}{t^*}-1\right)\left[\left(\dfrac{1-\pi_{1,0}}{1-\pi_{1,1}}\right)^{t^*}-\left(\dfrac{\pi_{1,0}}{\pi_{1,1}}\right)^{t^*} \right] < \dfrac{1-P_f}{1-P_d} \left[\left(\dfrac{\pi_{1,1}}{\pi_{1,0}}\right)\left(\dfrac{\pi_{1,0}}
{\pi_{1,1}}\right)^{t^*}-\left(\dfrac{1-\pi_{1,1}}{1-\pi_{1,0}}\right)\left(\dfrac{1-\pi_{1,0}}{1-\pi_{1,1}}\right)^{t^*} \right] \\
&\Leftrightarrow  &
\left(\dfrac{1-\pi_{1,0}}{1-\pi_{1,1}}\right)^{t^*}\left[\dfrac{1-P_f}{1-P_d}\left(\dfrac{1-\pi_{1,1}}{1-\pi_{1,0}}\right)+\left(\dfrac{1}{t^*}-1\right)\right] < \left(\dfrac{\pi_{1,0}}{\pi_{1,1}}\right)^{t^*}\left[\dfrac{1-P_f}{1-P_d}\left(\dfrac{\pi_{1,1}}{\pi_{1,0}}\right)+\left(\dfrac{1}{t^*}-1\right)\right]\\
&\Leftrightarrow  &
\left(\dfrac{(1/\pi_{1,0})-1}{(1/\pi_{1,1})-1}\right)^{t^*}\left[\dfrac{1-P_f}{1-P_d}\left(\dfrac{1-\pi_{1,1}}{1-\pi_{1,0}}\right)+\left(\dfrac{1}{t^*}-1\right)\right] < \left[\dfrac{1-P_f}{1-P_d}\left(\dfrac{\pi_{1,1}}{\pi_{1,0}}\right)+\left(\dfrac{1}{t^*}-1\right)\right].
\end{eqnarray*}
Using the result from \eqref{chernoffK}, the above equation can be written as
\begin{equation*}
\dfrac{\ln (\pi_{1,1}/\pi_{1,0})}{\ln \left(\dfrac{(1-\pi_{1,0})}{(1-\pi_{1,1})}\right)} \left(\dfrac{\pi_{1,1}}{1-\pi_{1,1}}\right)< \dfrac{\left[\dfrac{1-P_f}{1-P_d}\left(\dfrac{\pi_{1,1}}{\pi_{1,0}}\right)+\left(\dfrac{1}{t^*}-1\right)\right]}{\left[\dfrac{1-P_f}{1-P_d}\left(\dfrac{1-\pi_{1,1}}{1-\pi_{1,0}}\right)+\left(\dfrac{1}{t^*}-1\right)\right]}.
\end{equation*}
Using the fact that $G=\dfrac{\ln (\pi_{1,1}/\pi_{1,0})}{\ln \left(\dfrac{(1-\pi_{1,0})}{(1-\pi_{1,1})}\right)}$, we get
\begin{equation*}
G< \dfrac{\left[\dfrac{1-P_f}{1-P_d}\left(\dfrac{1}{\pi_{1,0}}\right)+\left(\dfrac{1}{t^*}-1\right)\left(\dfrac{1}{\pi_{1,1}}\right)\right]}{\left[\dfrac{1-P_f}{1-P_d}\left(\dfrac{1}{1-\pi_{1,0}}\right)+\left(\dfrac{1}{t^*}-1\right)\left(\dfrac{1}{1-\pi_{1,1}}\right)\right]}.
\end{equation*}
After some simplification, the above condition can be written as
\begin{eqnarray*}
&& \dfrac{1-P_f}{1-P_d}\left[\dfrac{G}{1-\pi_{1,0}}-\dfrac{1}{\pi_{1,0}}\right]< \left(\dfrac{1}{t^*}-1\right)\left[\dfrac{1}{\pi_{1,1}}-\dfrac{G}{1-\pi_{1,1}}\right]\\
&\Leftrightarrow  &
\left(\dfrac{1-P_f}{1-P_d}\right)\left(\dfrac{\pi_{1,1}}{\pi_{1,0}}\right)\left(\dfrac{1-\pi_{1,1}}{1-\pi_{1,0}}\right)[\pi_{1,0}(G+1)-1]< \left(\dfrac{1}{t^*}-1\right)[1-\pi_{1,1}(G+1)]
\end{eqnarray*}
\begin{equation}
\label{bk}
\dfrac{1}{t^*}[\pi_{1,1}(G+1)-1]<\left(\dfrac{1-P_f}{1-P_d}\right)\left(\dfrac{\pi_{1,1}}{\pi_{1,0}}\right)\left(\dfrac{1-\pi_{1,1}}{1-\pi_{1,0}}\right)[1-\pi_{1,0}(G+1)]+[\pi_{1,1}(G+1)-1].
\end{equation}

Notice that, in the above equation 
\begin{equation}
\label{inter}
\pi_{1,1}(G+1)\geq 1\; \text{and} \;\pi_{1,0}(G+1)\leq 1
\end{equation}
 or equivalently $\dfrac{1-\pi_{1,1}}{\pi_{1,1}} \leq G\leq \dfrac{1-\pi_{1,0}}{\pi_{1,0}}$. The second inequality in \eqref{inter} follows from the fact that $\ln\left(\dfrac{\pi_{1,1}}{\pi_{1,0}}\right)\geq \ln\left(\frac{1-\pi_{1,0}}{1-\pi_{1,1}}\right)\dfrac{1-\pi_{1,1}}{\pi_{1,1}}$. Using logarithm inequality, we have
$\ln\left(\dfrac{\pi_{1,1}}{\pi_{1,0}}\right)\geq \left(1-\dfrac{\pi_{1,0}}{\pi_{1,1}}\right)=
\left(\dfrac{1-\pi_{1,1}}{\pi_{1,1}}\right)\left(\dfrac{1-\pi_{1,0}}{1-\pi_{1,1}}-1\right)\geq
 \ln\left(\dfrac{1-\pi_{1,0}}{1-\pi_{1,1}}\right)\dfrac{1-\pi_{1,1}}{\pi_{1,1}}$.
 Similarly, to show that the second inequality in \eqref{inter} is true we show
 $\ln\left(\dfrac{\pi_{1,1}}{\pi_{1,0}}\right)\leq \ln\left(\dfrac{1-\pi_{1,0}}{1-\pi_{1,1}}\right)\dfrac{1-\pi_{1,0}}{\pi_{1,0}}$. Using logarithm inequality,
$\ln\left(\dfrac{\pi_{1,1}}{\pi_{1,0}}\right)\leq \left(\dfrac{\pi_{1,1}}{\pi_{1,0}}-1\right)=
\left(\dfrac{1-\pi_{1,0}}{\pi_{1,0}}\right)\left(1-\dfrac{1-\pi_{1,1}}{1-\pi_{1,0}}\right)\leq
 \ln\left(\dfrac{1-\pi_{1,0}}{1-\pi_{1,1}}\right)\dfrac{1-\pi_{1,0}}{\pi_{1,0}}$.
 Using these results we can then write \eqref{bk} in the form below,
\begin{equation*}
\dfrac{[\pi_{1,1}(G+1)-1]}{\left(\dfrac{1-P_f}{1-P_d}\right)\left(\dfrac{\pi_{1,1}}{\pi_{1,0}}\right)\left(\dfrac{1-\pi_{1,1}}{1-\pi_{1,0}}\right)[1-\pi_{1,0}(G+1)]+[\pi_{1,1}(G+1)-1]}< {t^*}
\end{equation*}
\begin{equation}
\label{chernoffCOND1}
\Leftrightarrow
\dfrac{1}{\left(\dfrac{1-P_f}{1-P_d}\right)\left(\dfrac{\pi_{1,1}}{\pi_{1,0}}\right)\left(\dfrac{1-\pi_{1,1}}{1-\pi_{1,0}}\right)\dfrac{[1-\pi_{1,0}(G+1)]}{[\pi_{1,1}(G+1)-1]}+1}< {t^*}.
\end{equation}
Similarly, the left hand side inequality in \eqref{chernoffcond} can be written as,
\begin{eqnarray*} 
&&
\dfrac{P_f}{P_d} \left[\left(\dfrac{\pi_{1,1}}{\pi_{1,0}}\right)^{1-{t^*}}-\left(\dfrac{1-\pi_{1,1}}{1-\pi_{1,0}}\right)^{1-{t^*}} \right]< \left(\dfrac{1}{{t^*}}-1\right)\left[\left(\dfrac{1-\pi_{1,0}}{1-\pi_{1,1}}\right)^{t^*}-\left(\dfrac{\pi_{1,0}}{\pi_{1,1}}\right)^{t^*} \right]   \\
&\Leftrightarrow  &
 \dfrac{P_f}{P_d} \left[\left(\dfrac{\pi_{1,1}}{\pi_{1,0}}\right)\left(\dfrac{\pi_{1,0}}
{\pi_{1,1}}\right)^{t^*}-\left(\dfrac{1-\pi_{1,1}}{1-\pi_{1,0}}\right)\left(\dfrac{1-\pi_{1,0}}{1-\pi_{1,1}}\right)^{t} \right] < \left(\dfrac{1}{t^*}-1\right)\left[\left(\dfrac{1-\pi_{1,0}}{1-\pi_{1,1}}\right)^{t^*}-\left(\dfrac{\pi_{1,0}}{\pi_{1,1}}\right)^{t^*} \right]\\
&\Leftrightarrow  &
  \left(\dfrac{\pi_{1,0}}{\pi_{1,1}}\right)^{t}\left[\dfrac{P_f}{P_d}\left(\dfrac{\pi_{1,1}}{\pi_{1,0}}\right)+\left(\dfrac{1}{t^*}-1\right)\right] < \left(\dfrac{1-\pi_{1,0}}{1-\pi_{1,1}}\right)^{t}\left[\dfrac{P_f}{P_d}\left(\dfrac{1-\pi_{1,1}}{1-\pi_{1,0}}\right)+\left(\dfrac{1}{t^*}-1\right)\right]
 \\
 &\Leftrightarrow  &
 \left[\dfrac{P_f}{P_d}\left(\dfrac{\pi_{1,1}}{\pi_{1,0}}\right)+\left(\dfrac{1}{t^*}-1\right)\right] < \left(\dfrac{(1/\pi_{1,0})-1}{(1/\pi_{1,1})-1}\right)^{t^*}\left[\dfrac{P_f}{P_d}\left(\dfrac{1-\pi_{1,1}}{1-\pi_{1,0}}\right)+\left(\dfrac{1}{t^*}-1\right)\right].
\end{eqnarray*}
Using the results from \eqref{chernoffK}, the above equation can be written as,
\begin{equation*}
 \dfrac{\left[\dfrac{P_f}{P_d}\left(\dfrac{\pi_{1,1}}{\pi_{1,0}}\right)+\left(\dfrac{1}{t^*}-1\right)\right]}{\left[\dfrac{P_f}{P_d}\left(\dfrac{1-\pi_{1,1}}{1-\pi_{1,0}}\right)+\left(\dfrac{1}{t^*}-1\right)\right]}< \dfrac{\ln (\pi_{1,1}/\pi_{1,0})}{\ln \left(\dfrac{(1-\pi_{1,0})}{(1-\pi_{1,1})}\right)} \left(\dfrac{\pi_{1,1}}{1-\pi_{1,1}}\right).
\end{equation*}
Lets denote $G=\dfrac{\ln (\pi_{1,1}/\pi_{1,0})}{\ln \left(\dfrac{(1-\pi_{1,0})}{(1-\pi_{1,1})}\right)}$, we get
\begin{equation*}
 \dfrac{\left[\dfrac{P_f}{P_d}\left(\dfrac{1}{\pi_{1,0}}\right)+\left(\dfrac{1}{t^*}-1\right)\left(\dfrac{1}{\pi_{1,1}}\right)\right]}{\left[\dfrac{P_f}{P_d}\left(\dfrac{1}{1-\pi_{1,0}}\right)+\left(\dfrac{1}{t^*}-1\right)\left(\dfrac{1}{1-\pi_{1,1}}\right)\right]}< G.
\end{equation*}
After some simplification the above condition can be written as,
\begin{eqnarray*}
&&
 \left(\dfrac{1}{t^*}-1\right)\left[\dfrac{1}{\pi_{1,1}}-\dfrac{G}{1-\pi_{1,1}}\right] < \dfrac{P_f}{P_d}\left[\dfrac{G}{1-\pi_{1,0}}-\dfrac{1}{\pi_{1,0}}\right]\ \\
&\Leftrightarrow  &
 \left(\dfrac{1}{t^*}-1\right)[1-\pi_{1,1}(G+1)]< \left(\dfrac{P_f}{P_d}\right)\left(\dfrac{\pi_{1,1}}{\pi_{1,0}}\right)\left(\dfrac{1-\pi_{1,1}}{1-\pi_{1,0}}\right)[\pi_{1,0}(G+1)-1]\\
 &\Leftrightarrow  &
\left(\dfrac{P_f}{P_d}\right)\left(\dfrac{\pi_{1,1}}{\pi_{1,0}}\right)\left(\dfrac{1-\pi_{1,1}}{1-\pi_{1,0}}\right)[1-\pi_{1,0}(G+1)]+[\pi_{1,1}(G+1)-1]< \dfrac{1}{t^*}[\pi_{1,1}(G+1)-1].
\end{eqnarray*}
Using \eqref{inter}, the condition can be written as 
\begin{equation*}
{t^*} < \dfrac{[\pi_{1,1}(G+1)-1]}{\left(\dfrac{P_f}{P_d}\right)\left(\dfrac{\pi_{1,1}}{\pi_{1,0}}\right)\left(\dfrac{1-\pi_{1,1}}{1-\pi_{1,0}}\right)[1-\pi_{1,0}(G+1)]+[\pi_{1,1}(G+1)-1]}
\end{equation*}
\begin{equation}
\label{chernoffCOND2}
{t^*}< \dfrac{1}{\left(\dfrac{P_f}{P_d}\right)\left(\dfrac{\pi_{1,1}}{\pi_{1,0}}\right)\left(\dfrac{1-\pi_{1,1}}{1-\pi_{1,0}}\right)\dfrac{[1-\pi_{1,0}(G+1)]}{[\pi_{1,1}(G+1)-1]}+1}.
\end{equation}
Now from \eqref{chernoffCOND1} and \eqref{chernoffCOND2}, Lemma \ref{Lemma-6} is true if
\begin{equation}
\footnotesize
\label{chernoffCOND}
A=\dfrac{1}{\left(\dfrac{1-P_f}{1-P_d}\right)\left(\dfrac{\pi_{1,1}}{\pi_{1,0}}\right)\left(\dfrac{1-\pi_{1,1}}{1-\pi_{1,0}}\right)\dfrac{[1-\pi_{1,0}(G+1)]}{[\pi_{1,1}(G+1)-1]}+1}
< {t^*}< \dfrac{1}{\dfrac{P_f}{P_d}\left(\dfrac{\pi_{1,1}}{\pi_{1,0}}\right)\left(\dfrac{1-\pi_{1,1}}{1-\pi_{1,0}}\right)\dfrac{[1-\pi_{1,0}(G+1)]}{[\pi_{1,1}(G+1)-1]}+1}=B.
\end{equation}
Next, we show that, the optimal $t^*$ is with in the region $(A,B)$. We start from the inequality
\begin{eqnarray*}
&&
\dfrac{P_f}{P_d}\dfrac{\pi_{1,1}}{\pi_{1,0}}<1<\dfrac{1-P_f}{1-P_d}\dfrac{1-\pi_{1,1}}{1-\pi_{1,0}}\\
&\Leftrightarrow&
\dfrac{P_f}{P_d}\dfrac{\pi_{1,1}}{\pi_{1,0}}\dfrac{[1-\pi_{1,0}(G+1)]}{[\pi_{1,1}(G+1)-1]}<\dfrac{1-\pi_{1,0}(G+1)}{\pi_{1,1}(G+1)-1}<\dfrac{1-P_f}{1-P_d}\dfrac{1-\pi_{1,1}}{1-\pi_{1,0}}\dfrac{[1-\pi_{1,0}(G+1)]}{[\pi_{1,1}(G+1)-1]}
\end{eqnarray*}
Let us denote, $Y=\left(\dfrac{\pi_{1,1}}{\pi_{1,0}}\right)\left(\dfrac{1-\pi_{1,1}}{1-\pi_{1,0}}\right)\dfrac{[1-\pi_{1,0}(G+1)]}{[\pi_{1,1}(G+1)-1]}$, then the above condition can be written as,
\begin{eqnarray}
&&
\dfrac{P_f}{P_d}Y\dfrac{1-\pi_{1,0}}{1-\pi_{1,1}}<\dfrac{1-\pi_{1,0}(G+1)}{\pi_{1,1}(G+1)-1}<\dfrac{1-P_f}{1-P_d}Y\dfrac{\pi_{1,0}}{\pi_{1,1}}\label{above}
\end{eqnarray}
Next, we use the log inequality, $\dfrac{x-1}{x}<\ln(x)<(x-1),\; \forall x>0$, to derive further results. Let us focus our attention to the left hand side inequality in \eqref{above} 
\begin{eqnarray}
&&
\dfrac{P_f}{P_d}Y\dfrac{1-\pi_{1,0}}{1-\pi_{1,1}}<\dfrac{1-\pi_{1,0}(G+1)}{\pi_{1,1}(G+1)-1} \nonumber \\
&\Leftrightarrow&
P_f Y\left[\dfrac{G\pi_{1,1}}{1-\pi_{1,1}}-1\right]<P_d \left[1-\dfrac{G\pi_{1,0}}{1-\pi_{1,0}}\right]\nonumber \\
&\Leftrightarrow&
P_f Y{\ln\left(G \dfrac{\pi_{1,1}}{1-\pi_{1,1}}\right)} <P_d {\ln\left(\dfrac{1}{G}\dfrac{1-\pi_{1,0}}{\pi_{1,0}}\right)}\label{above1}
\end{eqnarray}
Now, let us focus our attention to the right hand side inequality in \eqref{above} 
\begin{eqnarray}
&&
\dfrac{1-\pi_{1,0}(G+1)}{\pi_{1,1}(G+1)-1}<\dfrac{1-P_f}{1-P_d}Y\dfrac{\pi_{1,0}}{\pi_{1,1}} \nonumber \\
&\Leftrightarrow&
(1-P_d)\left(\dfrac{1-\pi_{1,0}}{G \pi_{1,0}}-1\right) 
<(1-P_f)Y{\left(1-\dfrac{1-\pi_{1,1}}{G\pi_{1,1}}\right)}\nonumber \\
&\Leftrightarrow&
(1-P_d){\ln\left(\dfrac{1}{G}\dfrac{1-\pi_{1,0}}{\pi_{1,0}}\right)} <(1-P_f)Y{\ln\left(G \frac{\pi_{1,1}}{1-\pi_{1,1}}\right)}\label{above2}
\end{eqnarray}
Now using the results from \eqref{above1} and \eqref{above2}, we can deduce that 
\begin{eqnarray}
&&
\left(\dfrac{P_f}{P_d}\right)Y<\dfrac{\ln\left(\dfrac{1}{G}\dfrac{1-\pi_{1,0}}{\pi_{1,0}}\right)}{\ln\left(G \dfrac{\pi_{1,1}}{1-\pi_{1,1}}\right)} <\left(\dfrac{1-P_f}{1-P_d}\right)Y \nonumber \\
&\Leftrightarrow&
\dfrac{1}{1+\left(\dfrac{1-P_f}{1-P_d}\right)Y}
< \dfrac{1}{1+\dfrac{\ln\left(\dfrac{1}{G}\dfrac{1-\pi_{1,0}}{\pi_{1,0}}\right)}{\ln\left(G \dfrac{\pi_{1,1}}{1-\pi_{1,1}}\right)}}< \dfrac{1}{1+\left(\dfrac{P_f}{P_d}\right)Y}\label{above3}
\end{eqnarray}
which is true from the fact that for $a>0,b>0$, $\dfrac{1}{1+a}<\dfrac{1}{1+b}$ iff $b<a$. Next, observe that, $t^*$ as given in \eqref{t} can be written as 
\begin{equation*}
t^*=\dfrac{\ln \left(G \dfrac{\pi_{1,1}}{1-\pi_{1,1}}\right)}{\ln \left(\dfrac{(1/\pi_{1,0})-1}{(1/\pi_{1,1})-1}\right)}=\dfrac{\ln(G) + \ln \left(\dfrac{\pi_{1,1}}{1-\pi_{1,1}}\right)}{\ln  \left(\dfrac{\pi_{1,1}}{\pi_{1,0}}\right)+\ln\left(\dfrac{1-\pi_{1,0}}{1-\pi_{1,1}}\right)}.
\end{equation*}
Observe that, $\ln \left(G \dfrac{\pi_{1,1}}{1-\pi_{1,1}}\right)\geq 0$ and $\ln\left(\dfrac{1}{G}\dfrac{1-\pi_{1,0}}{\pi_{1,0}}\right)\geq 0$ or equivalently $ \left(G \dfrac{\pi_{1,1}}{1-\pi_{1,1}}\right)\geq 1$ and $\left(\dfrac{1}{G}\dfrac{1-\pi_{1,0}}{\pi_{1,0}}\right)\geq 1$ from \eqref{inter}. Now,
\begin{equation*}
t^*=\dfrac{1}{1+\dfrac{\ln  \left(\dfrac{\pi_{1,1}}{\pi_{1,0}}\right)+\ln\left(\dfrac{1-\pi_{1,0}}{1-\pi_{1,1}}\right)-\ln(G) - \ln \left(\dfrac{\pi_{1,1}}{1-\pi_{1,1}}\right)}{\ln(G) + \ln \left(\dfrac{\pi_{1,1}}{1-\pi_{1,1}}\right)}}=\dfrac{1}{1+\dfrac{\ln\left(\dfrac{1}{G}\dfrac{1-\pi_{1,0}}{\pi_{1,0}}\right)}{\ln\left(G \dfrac{\pi_{1,1}}{1-\pi_{1,1}}\right)}}. 
\end{equation*}
Which along with \eqref{above3} implies that
\begin{eqnarray*}
&&
\dfrac{1}{1+\left(\dfrac{1-P_f}{1-P_d}\right)Y}
< t^*< \dfrac{1}{1+\left(\dfrac{P_f}{P_d}\right)Y}
\end{eqnarray*}
or in other words, $A< t^*<B$. This completes our proof.

\bibliographystyle{IEEEtran}
\bibliography{Conf,Book,Journal}

\begin{thebibliography}{10}
\providecommand{\url}[1]{#1}
\csname url@samestyle\endcsname
\providecommand{\newblock}{\relax}
\providecommand{\bibinfo}[2]{#2}
\providecommand{\BIBentrySTDinterwordspacing}{\spaceskip=0pt\relax}
\providecommand{\BIBentryALTinterwordstretchfactor}{4}
\providecommand{\BIBentryALTinterwordspacing}{\spaceskip=\fontdimen2\font plus
\BIBentryALTinterwordstretchfactor\fontdimen3\font minus
  \fontdimen4\font\relax}
\providecommand{\BIBforeignlanguage}[2]{{%
\expandafter\ifx\csname l@#1\endcsname\relax
\typeout{** WARNING: IEEEtran.bst: No hyphenation pattern has been}%
\typeout{** loaded for the language `#1'. Using the pattern for}%
\typeout{** the default language instead.}%
\else
\language=\csname l@#1\endcsname
\fi
#2}}
\providecommand{\BIBdecl}{\relax}
\BIBdecl

\bibitem{Varshney}
P.~K. Varshney, \emph{Distributed Detection and Data Fusion}.\hskip 1em plus
  0.5em minus 0.4em\relax New York:Springer-Verlag, 1997.

\bibitem{Viswanathan}
R.~Viswanathan and P.~K. Varshney, ``{Distributed detection with multiple
  sensors: Part I - Fundamentals},'' \emph{Proc. IEEE}, vol.~85, no.~1, pp. 54
  --63, Jan 1997.

\bibitem{veer}
V.~Veeravalli and P.~K. Varshney, ``{Distributed inference in wireless sensor
  networks},'' \emph{Philosophical Transactions of the Royal Society A:
  Mathematical, Physical and Engineering Sciences}, vol. 370, pp. 100--117,
  2012.

\bibitem{Lamport}
\BIBentryALTinterwordspacing
L.~Lamport, R.~Shostak, and M.~Pease, ``{The Byzantine Generals Problem},''
  \emph{ACM Trans. Program. Lang. Syst.}, vol.~4, no.~3, pp. 382--401, Jul.
  1982. [Online]. Available: \url{http://doi.acm.org/10.1145/357172.357176}
\BIBentrySTDinterwordspacing

\bibitem{frag}
A.~Fragkiadakis, E.~Tragos, and I.~Askoxylakis, ``{A Survey on Security Threats
  and Detection Techniques in Cognitive Radio Networks},'' \emph{IEEE
  Communications Surveys Tutorials}, vol.~15, no.~1, pp. 428--445, 2013.

\bibitem{Rifa}
\BIBentryALTinterwordspacing
H.~Rif\`{a}-Pous, M.~J. Blasco, and C.~Garrigues, ``Review of robust
  cooperative spectrum sensing techniques for cognitive radio networks,''
  \emph{Wirel. Pers. Commun.}, vol.~67, no.~2, pp. 175--198, Nov. 2012.
  [Online]. Available: \url{http://dx.doi.org/10.1007/s11277-011-0372-x}
\BIBentrySTDinterwordspacing

\bibitem{Marano}
S.~Marano, V.~Matta, and L.~Tong, ``{Distributed Detection in the Presence of
  Byzantine Attacks},'' \emph{IEEE Trans. Signal Process.}, vol.~57, no.~1, pp.
  16 --29, Jan. 2009.

\bibitem{Rawat}
A.~Rawat, P.~Anand, H.~Chen, and P.~Varshney, ``{Collaborative Spectrum Sensing
  in the Presence of Byzantine Attacks in Cognitive Radio Networks},''
  \emph{IEEE Trans. Signal Process.}, vol.~59, no.~2, pp. 774 --786, Feb 2011.

\bibitem{bhavyaj}
B.~Kailkhura, S.~Brahma, Y.~S. Han, and P.~K. Varshney, ``{Distributed
  Detection in Tree Topologies With Byzantines},'' \emph{IEEE Trans. Signal
  Process.}, vol.~62, pp. 3208--3219, June 2014.

\bibitem{Kailkhura2013}
B.~Kailkhura, S.~Brahma, and P.~K. Varshney, ``{Optimal Byzantine Attack on
  Distributed Detection in Tree based Topologies},'' in \emph{Proc.
  International Conference on Computing, Networking and Communications
  Workshops (ICNC-2013)}, San Diego, CA, January 2013, pp. 227--231.

\bibitem{Kailkhura}
B.~Kailkhura, S.~Brahma, Y.~S. Han, and P.~K. Varshney, ``{Optimal Distributed
  Detection in the Presence of Byzantines},'' in \emph{Proc. The 38th
  International Conference on Acoustics, Speech, and Signal Processing (ICASSP
  2013)}, Vancouver, Canada, May 2013.

\bibitem{aditya}
A.~Vempaty, K.~Agrawal, H.~Chen, and P.~K. Varshney, ``{Adaptive learning of
  Byzantines' behavior in cooperative spectrum sensing},'' in \emph{Proc. IEEE
  Wireless Comm. and Networking Conf. (WCNC)}, march 2011, pp. 1310 --1315.

\bibitem{a1}
R.~Chen, J.-M. Park, and K.~Bian, ``Robust distributed spectrum sensing in
  cognitive radio networks,'' in \emph{Proc. 27th Conf. Comput. Commun.,
  Phoenix, AZ}, 2008, pp. 1876--1884.

\bibitem{a2}
E.~Soltanmohammadi, M.~Orooji, and M.~Naraghi-Pour, ``{Decentralized Hypothesis
  Testing in Wireless Sensor Networks in the Presence of Misbehaving Nodes},''
  \emph{IEEE Trans. Inf. Forensics Security}, vol.~8, no.~1, pp. 205--215,
  2013.

\bibitem{covert}
B.~Kailkhura, Y.~Han, S.~Brahma, and P.~Varshney, ``{On Covert Data
  Falsification Attacks on Distributed Detection Systems},'' in
  \emph{Communications and Information Technologies (ISCIT), 2013 13th
  International Symposium on}, Sept 2013, pp. 412--417.

\bibitem{tsit}
J.~N. Tsitsiklis, ``{Decentralized Detection by a Large Number of Sensors*},''
  \emph{Math. control, Signals, and Systems}, vol.~1, pp. 167--182, 1988.

\bibitem{HCHER}
H.~Chernoff, ``{A Measure of Asymptotic Efficiency for Tests of a Hypothesis
  Based on the sum of Observations},'' \emph{The Annals of Mathematical
  Statistics}, vol.~23, pp. 493--507, December 1952.

\end{thebibliography}

\end{document}